\documentclass[conference]{IEEEtran}  
\IEEEoverridecommandlockouts
\usepackage{amsmath,amssymb,amsthm,mathrsfs,amsfonts,dsfont,stmaryrd}
\usepackage{graphicx}  
\allowdisplaybreaks[4]
\newtheorem{theorem}{Theorem}
\newtheorem{definition}{Definition}

\newtheorem{lemma}{Lemma}

\newcommand{\mkv}{-\!\!\!\!\minuso\!\!\!\!-}

\DeclareMathOperator*{\dprime}{\prime \prime}
\DeclareMathOperator*{\trprime}{\prime\prime\prime}

\title{Multi-Kernel Polar Codes: Proof of Polarization  \\
 and Error Exponents}

\author{\IEEEauthorblockN{Meryem Benammar, Valerio Bioglio, Fr\'ed\'eric Gabry\thanks{These results were derived when F. Gabry was still with the Mathematical and Algorithmic Sciences Lab of Huawei.}, Ingmar Land}
\IEEEauthorblockA{Mathematical and Algorithmic Sciences Lab\\ Paris Research Center, Huawei Technologies France SASU\\
Email: $\{$meryem.benammar,valerio.bioglio,frederic.gabry,ingmar.land$\}$@huawei.com}\vspace{-2mm}} 

\begin{document}
\maketitle 
\begin{abstract}
In this paper, we investigate a novel family of polar codes based on multi-kernel constructions, proving that this construction actually polarizes. 
To this end, we derive a new and more general proof of polarization, which gives sufficient conditions for kernels to polarize. 
Finally, we derive the convergence rate of the multi-kernel construction and relate it to the convergence rate of each of the constituent kernels. 
\end{abstract}

\section{Introduction}
Channel polarization is a novel technique to create capacity-achieving codes over various channels \cite{polar}. In its original construction, a polar code is generated by a sub-matrix of the transformation matrix $T_2^{\otimes n}$ with the binary kernel defined by $T_2 = \left(\begin{array}{cc}
 1 & 0 \\ 
 1 & 1 \\ 
 \end{array} \right)$. 
While the polarization of the Kronecker powers of the $T_2$ kernel is presented in \cite{polar}, the proof has been generalized in \cite{exp_urbanke} for the Kronecker power of larger binary kernels. 
Hereafter, various kernels have been proposed, along with their convergence rates \cite{kernel_presman}. 

Recently, a novel family of polar codes, where binary kernels of different sizes are mixed, has been proposed in \cite{mk_arxiv}. 
These codes, coined as \textit{multi-kernel polar codes}, make it possible to construct polar codes of any block lengths, not limited to powers of integers as when using a single kernel. 
In \cite{mk_arxiv}, authors conjecture that the proposed mixed construction polarizes; this conjecture is confirmed by the density evolution algorithm used to calculate the bit reliabilities \cite{DE_mori}, however, a rigorous convergence proof is missing.  

In this paper, we fill the gap left in \cite{mk_arxiv} by presenting a proof of the polarization of multi-kernel polar codes. 
To this end, we describe the polarization of any kernel through an inequality related to the reliability transitions in every step of the polarization. 
This inequality is different from the ones presented in other proofs \cite{polar,exp_urbanke}; however, we will show that this inequality holds for Arikan's original kernel, along with other selected kernels. 
Besides, we derive the convergence rate of the obtained multi-kernel polar code based on convergence rate of each of the constituent kernels. 

This paper is organized as follows. 
Section II introduces the information theoretic model and reviews the code construction of multi-kernel polar codes. Section III presents the main result of this paper which consists in the proof of polarization, while Section IV derives the convergence rates (or error exponents) of the resulting multi-kernel polar code. 
 
\section{System model and code definition}
In this section, we introduce the fundamental definitions related to polar coding and the underlying information theoretic model. 
Moreover, we briefly present the multi-kernel construction of polar codes. 

\subsection{Channel model}
Let $ \mathcal{W} :  \mathcal{X} \rightarrow \mathcal{Y} $  be a discrete input channel defined by its associated probability mass function (pmf) $W(y|x)$, and let $\mathcal{W}^{(N)} : \mathcal{X}^N  \rightarrow \mathcal{Y}^N $ be its $N$-th memoryless extension with associated pmf $W^{(N)}(y_1^N|x_1^N)$, 
\begin{equation}
  W^{(N)}(y_1^N|x_1^N) =  \prod_{i=1}^N  W(y_i|x_i) . 
\end{equation}
 
Let $U_1^N \triangleq (U_1, \dots, U_N)$ be $N$ auxiliary random variables satisfying the Markov chain $(U_1, \dots, U_N) \mkv X_i \mkv Y_i $ for all $i \in [1:N]$, where $[1:N]$ represents the set of integers from $1$ to $N$. 
In the following, the input alphabet $\mathcal{X}$ and the respective auxiliary alphabets $\mathcal{U}_j$ are all binary, i.e., for all $i \in [1:N]$, $ \mathcal{X} = \mathcal{U}_j = \{0,1\} $. 
The auxiliary variables, or \textit{bit components}, $U_i$ are all pairwise independent $Bern(\frac{1}{2})$ variables.  
  
\subsection{Channel polarization}
A polar code of length $N$ is a linear block code which maps the bits $u_1^N=(u_1, \dots, u_N)$ into the channel input array $x_1^N=(x_1, \dots , x_N) $ through a linear invertible mapping, i.e., 
\begin{equation}
 x_1^n = u_1^n \cdot G_N\ , 
\end{equation}
where $G_N$ is depicted as the \textit{transformation matrix} of the polar code. 
Since the bits $(U_1,\dots, U_N)$ are independent $Bern(\frac{1}{2})$ distributed and $G_N$ is invertible, then the channel inputs $(X_1, \dots, X_N)$ are also independent $Bern(\frac{1}{2})$ distributed. 
This yields to the \textit{information conservation} principle 
\begin{equation}
 I(U_1^N; Y_1^N) = I(X_1^N; Y_1^N) = N \cdot I(X;Y) ,  \label{eq-InformationConservation}
\end{equation}
which stems from that $\mathcal{W}^{(N)}$ is memoryless. 
In the following, since the input distribution $P_X$ is fixed, we will use $I(W) \triangleq I(X;Y)$ to denote the dependence of $I(X;Y) $ only on the channel $W$.

To introduce the \textit{polarization principle}, let us first use the independence of the bits $U_i$ to write the following
\begin{IEEEeqnarray}{rCl} 
I(U_1^N; Y_1^N) &=& \sum_{i=1}^N I(U_i; Y_1^N | U_1^{i-1}) \\
 &=&  \sum_{i=1}^N I(U_i; Y_1^N, U_1^{i-1}). \label{eq_VirtualChannels}
\end{IEEEeqnarray}
Let us then define the channel $\mathcal{W}^{(N)}_i$, with output $Z_i \triangleq (Y_1^N , U_1^{i-1})$ and associated pmf 
\begin{equation}\label{eq-new-channel}
  \mathcal{W}^{(N)}_i (z_i | u_i) = \mathcal{W}^{(N)}_i(y_1^N , u_1^{i-1}|u_i). 
\end{equation}
To polarize the channel $\mathcal{W}^{(N)}$ means to create $N$ virtual channels $\mathcal{W}^{(N)}_i$, each being either a \textit{degraded} version of $\mathcal{W}$ or an \textit{enhanced} version of it. 
By a proper choice of the transformation matrix $G_N$, as the one suggested by Arikan \cite{polar}, it can be shown that, as the code length $N$ goes to infinity, the resulting virtual channels $\mathcal{W}^{(N)}_i$ are either perfectly reliable channels or totally noisy channels, i.e., 
\begin{equation}
 \forall i \in [1:N] \quad I(U_i; Y_1^N , U_1^{i-1})  \rightarrow  0  \  \textrm{or} \  1. 
\end{equation}
Arikan showed in \cite{polar} that the fraction of perfectly reliable channels among all, i.e., the fraction of bits that can be transmitted reliably, is given by the mutual information $I(X;Y)$. 
  
\subsection{Multi-kernel code construction}
   
The transformation matrix $G_N$ of a polar code of length $N$ is defined by the $n$-fold Kronecker product of the binary \textit{kernel} $T_2 = \left(\begin{array}{cc}
    1 & 0 \\ 
    1 & 1 
\end{array} \right)$, namely $G_N = T_2^{\otimes n}$, where $N = 2^n$. 
Note that the structure of the transformation matrix renders it invertible, and that the admissible blocklengths $N$ are all powers of $2$, which might be an impediment for practical applications with arbitrary blocklengths. 
However, in \cite{exp_urbanke} authors prove that the kernel $T_2$ can be replaced by any polarizing kernel $T_l$ with dimension $l\times l$, leading to codes of blocklengths of the form $N =l^n$.
An example of such kernels, which we will resort to in the document, is the $T_3$ kernel given by: 
\begin{equation}\label{eq-T3Tkernels}
T_3 \triangleq \left(\begin{array}{ccc}
 1 & 1 & 1 \\ 
 1 & 0 & 1 \\ 
 0 & 1 & 1
 \end{array} \right) .
\end{equation} 

The multi-kernel polar code construction is introduced in \cite{mk_arxiv}, in which multiple binary kernels are used in the construction of the code. 
Consider to this end a collection of kernels $T_{l_1} , \dots, T_{l_m}$ where, for $j \in [1: m]$, $T_{l_j}$ is a binary matrix of dimension $l_j \times l_j$. 
A multi-kernel transformation matrix is constructed as the Kronecker product of these kernels,  
\begin{equation}
G_{N} = T_{l_1} \otimes T_{l_2}  \otimes \dots \otimes T_{l_m}. \label{eq_GeneratorMatrix}
\end{equation}
Note here that the size of the resulting code is $ {N} = \prod_{j=1}^m l_j$.
An $(N,K)$ multi-kernel polar code is defined by the transformation matrix $G_N$ in \eqref{eq_GeneratorMatrix}, and the information set $\mathcal{I}$ of size $||\mathcal{I}|| = K$, or conversely by the frozen set $\mathcal{F} = [1:N]\setminus \mathcal{I}$. 
For the encoding, each frozen bit is set to zero, i.e., $u_i = 0$ for $i\in \mathcal{F}$, while information is stored in the remaining bits, whose indices constitute the information set $\mathcal{I}$. 
Then, the channel input $x_1^N$ is obtained by $x_1^N  = u_1^N \cdot G_N$. 
In \cite{mk_arxiv}, authors conjecture that the Kronecker product of polarizing kernels result in a polarizing transformation matrix, and they calculate the reliabilities of the bits through density evolution \cite{DE_mori}. 
The information set $\mathcal{I}$ is then constituted by the $K$ bit positions with the highest reliability. 
In this paper, we present a prove of this conjecture, confirming the goodness of the multi-kernel construction in \cite{mk_arxiv}. 
 
Decoding of multi-kernel polar codes is performed similarly to Arikan's polar code, using successive cancellation decoding. 
At each step, a bit $u_i$ is decoded from the channel outputs $y_1^n$ using the previously made hard decisions on the bits $u_1^{i-1}$. 
In practice, the decoding is performed on the Tanner graph of the code, where each block of the decoder consists in the basic decoding operations of the kernel $T_{l_j}$ as explained in \cite{mk_arxiv}.
   
\section{Polarization of multi-kernel polar codes}
In this section, we prove the main result of the paper: the polarization of multi-kernel polar codes. 
The approach we adopt here is somewhat similar to the one proposed by Arikan in \cite{polar}, but differing in the last steps, where we prove that polarization is highly kernel dependent.  

 \subsection{Definitions} 
In the following, the sub-channels $W^{(N)}_i$ where $i \in [1:N]$ are defined as
\begin{equation}
 W^{(N)}_i \triangleq W_{b_1, \dots, b_m} 
\end{equation}
where the index $b_j$, for $j \in [1:m]$, takes values in $[1:l_j]$, constituting the \textit{mixed radix} decomposition of $i$ in the base $(l_{1},\dots, l_m)$. Due to the Kronecker construction of the transformation matrix $G_{N}$ in \eqref{eq_GeneratorMatrix}, the channel polarization occurs iteratively in that, at each step $m$, each of the previous channels $W_{b_1, \dots, b_{m-1}} $ is polarized into $l_m$ new channels $W_{b_1, \dots, b_{m-1}, 1}, \dots, W_{b_1, \dots, b_{m-1}, b_m}, \dots, W_{b_1, \dots, b_{m-1}, l_m}$ as shown in Figure \ref{Fig-Polarization}. 
\begin{figure}
\centering
\includegraphics[scale=0.55]{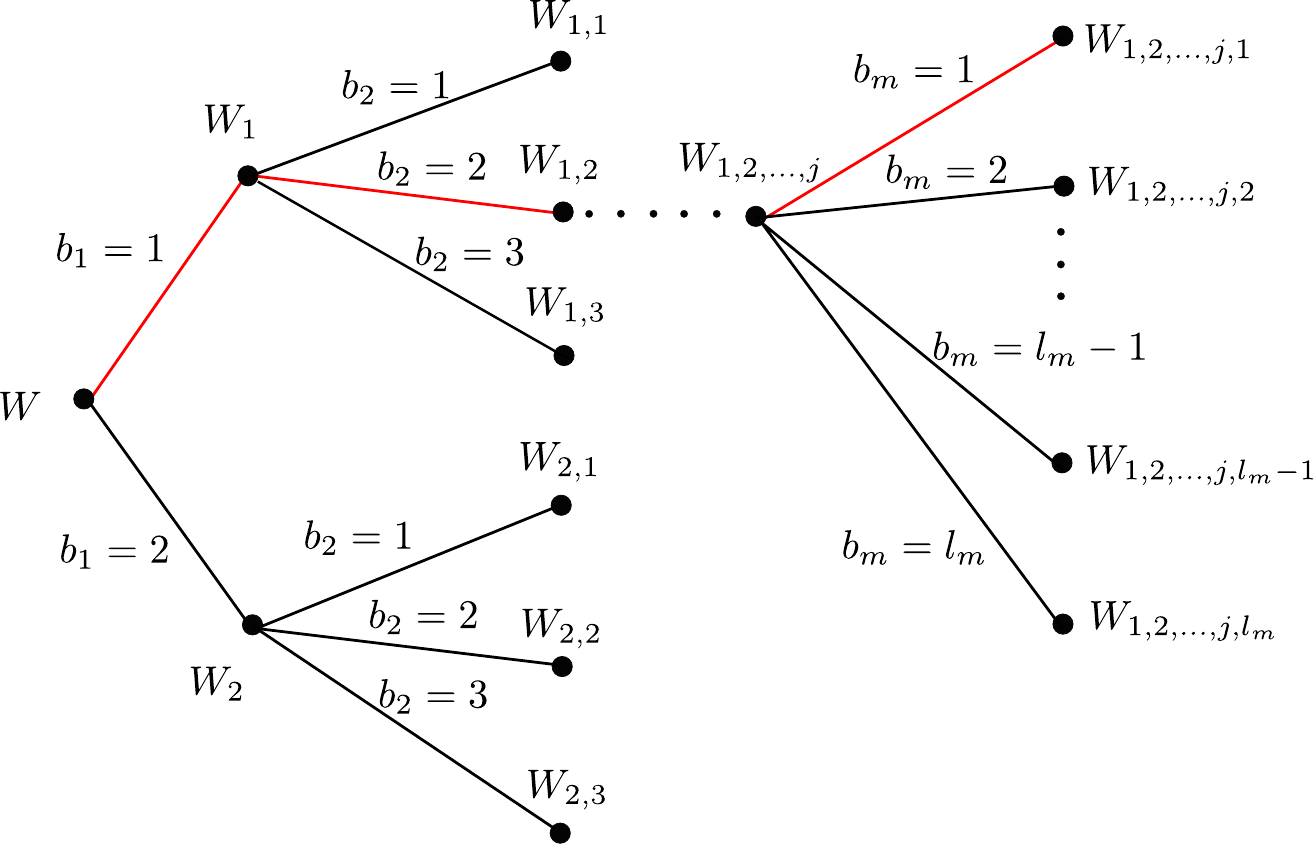}
\caption{Polarization tree of a kernel $T_2 \otimes T_3 \otimes ... \otimes T_{l_m} $}
\label{Fig-Polarization}
\end{figure}

Let us assume that the transitions on the polarization tree occur uniformly at random, i.e. the indices $(B_1,\dots ,B_m)$ are uniformly distributed, i.e., $\forall b_j \in [1:l_j], \  P_{B_j}(b_j) = \frac{1}{l_j}$. 
 The mutual information $I(W_{B_1,\dots, B_m})$ is then a random variable depending on the random variables $B_1, \dots, B_m$, which we will denote by  
\begin{equation}
I_m \triangleq I(W_{B_1, \dots, B_m})   
\end{equation}
where $I_0 \triangleq I(W)$ is a deterministic variable indicating the channel capacity. 
Finally, we give a formal definition of channel polarization for multi-kernel polar codes. 
\begin{definition}
\label{Def-pol}
A multi-kernel polar code is polarizing if 
\begin{equation*}
\forall \epsilon >0 ,  \lim_{m\to\infty} \mathds{P}(I_m \geq 1- \epsilon ) = 1 - \lim_{m\to\infty} \mathds{P}(I_m \leq\epsilon ) =  I(W). 
\end{equation*}
\end{definition}

\subsection{Proof of polarization}
Similarly to the case of $T_2$ kernels investigated by Arikan, the proof of polarization of multi-kernel polar codes is two-fold. 
First, we prove that the random process $\{I_m\}_m$ is converging to a random variable $I_\infty $ in probability. 
Then, we show that, if the kernels are selected properly, the distribution of $I_\infty $ is a Bernoulli distribution with probability $I(W)$, and $I_m$ follows the Definition \ref{Def-pol} . 

\begin{lemma}
\label{lem-mart}
The sequence $\{I_m\}_m$ is a bounded martingale with respect to $B_1, \dots, B_m$, and thus converges in probability to a variable $I_\infty$ whose expected value verifies $ \mathds{E}(I_\infty) = I_0 . $ 
\end{lemma}
\begin{proof}
To begin with, we have that $0 \leq I_m \leq 1 $ for all $m \in \mathds{N}$, which is due to the fact that the channel $W_{B_1,\dots,B_m}$ has binary inputs. 
Let $m \in \mathds{N}$, we have that
 \begin{IEEEeqnarray}{rCl} 
 && \mathds{E}_{B_{m+1}} (I_{m+1}| B_1, \dots, B_m) \\
 &=&   \mathds{E}_{B_{m+1}} \left(I(W_{ B_1, \dots, B_m , B_{m+1}} ) | B_1, \dots, B_m \right) \qquad \\
 &=&  \sum_{b= 1}^{l_{m+1}}  \mathds{P}(B_{m+1}= b)  I(W_{ B_1, \dots, B_m , b} ) \\
 &=& \frac{1}{l_{m+1}}  \sum_{b = 1}^{l_{m+1}}   I(W_{ B_1, \dots, B_m , b } ) \overset{(a)}{=}  I(W_{ B_1, \dots, B_m}) \label{eq-martingale}
 \end{IEEEeqnarray}
 where $(a)$ is a consequence the information conservation principle \eqref{eq-InformationConservation} and of \eqref{eq_VirtualChannels}.  
Finally, one can write that 
  \begin{IEEEeqnarray*}{rCl}
 \hspace{-3mm}&& \mathds{E}_{B_1, \dots, B_m, B_{m+1}} (I_{m+1})  \\
 &=&  \mathds{E}_{B_1, \dots, B_m} \left(  \mathds{E}_{B_{m+1}} \left(I(W_{ B_1, \dots, B_m, B_{m+1}}) | B_1, \dots, B_m \right) \right) \qquad   \\ 
 &\overset{(a)}{=}& \mathds{E}_{B_1, \dots, B_m}  (I_m) = I_m
 \end{IEEEeqnarray*}
 where $(a)$  is a consequence of \eqref{eq-martingale}. Thus, $\{\mathds{E} (I_m)\}_m$ is constant, and 
 \begin{equation}
   \mathds{E} (I_m) = \mathds{E} (I_{m-1}) =  \dots  = \mathds{E} (I_{0}) = I_0  \label{eq-expectation}.
 \end{equation}
Thus, $\{I_m\}_m$ is a bounded martingale, hence, uniformly integrable and thus convergent in probability to a random variable $I_\infty$ such that $
  \mathds{E} ( I_\infty ) =  I_0  , $
which is due to \eqref{eq-expectation}.
\end{proof}
The convergence of the martingale $\{I_m\}_m$ follows mainly from the information conservation property in \eqref{eq-InformationConservation} and is thus universal. 
On the other hand, polarization depends on to the kernel properties, as shown in the following Theorem. 

\begin{theorem}\label{th-inequality}
If for every kernel $T_{l}$ forming the transformation matrix of the code, and every past sequence $(b_1, \dots, b_{m-1})$, there exist $ \alpha ,\beta > 0 $ such that the following inequality holds:
\begin{IEEEeqnarray}{rCl} \label{eq-lemma2}
 && \forall b_m \in [1: l_m] \ , \ | I(W_{b_1, \dots, b_{m-1}, b_m}) - I(W_{b_1, \dots, b_{m-1}})   | \nonumber \\
 && \qquad \qquad \geq  I(W_{b_1, \dots, b_{m-1}})^{\alpha} \cdot \left(1 - I(W_{b_1, \dots, b_{m-1}})\right)^{\beta}, 
\end{IEEEeqnarray}
then $I_\infty$ is a Bernoulli distributed variable with 
\begin{equation}
 \mathds{P}(I_\infty = 1) = 1 - \mathds{P}(I_\infty = 0) = I_0. 
\end{equation}
\end{theorem}
\begin{proof}
Here we present a novel proof showing clearly the dependency of the polarization process on the choice of the kernel. 
Let $m \in \mathds{N}$ and let $(b_1,\dots , b_{m-1}) $ be the vector of past indices in the polarization tree. Under the assumptions of Theorem \ref{th-inequality}, we have that for all possible $b_m \in [1: l_m] $, 
\begin{IEEEeqnarray}{rCl}
 && | I(W_{b_1, \dots, b_{m-1}, b_m}) - I(W_{b_1, \dots, b_{m-1}})   | \nonumber \\
 && \qquad \qquad \geq  I(W_{b_1, \dots, b_{m-1}})^{\alpha} \cdot \left(1 - I(W_{b_1, \dots, b_{m-1}})\right)^{\beta}. 
\end{IEEEeqnarray}  
This implies that on the cylinder set $(B_1,\dots , B_m) $, the following inequality holds with probability $1$, i.e., 
\begin{equation}
 |  I_m - I_{m-1} | \geq I_{m-1}^{\alpha} \cdot (1- I_{m-1})^{\beta} \geq 0 . 
\end{equation}
Hence, since $\{I_{m}\}_m$ is a uniformly converging sequence, with limit $I_\infty$, and since the function $f: x \rightarrow x^\alpha (1-x)^\beta $ is continuous, then 
\begin{equation}
 0 \geq I_\infty^{\alpha}  \cdot (1- I_\infty)^{\beta}  \geq 0   , 
\end{equation}
 which in turn implies that $I_\infty = 1 \  \textrm{or} \ 0$. 
\end{proof}
What we presented can be seen as an alternative to the original proof of  polarization made by Arikan and then extended in \cite{exp_urbanke} to arbitrary kernels. 
Even if this inequality may seem a bit restrictive, it is verified for a big family of kernels; in the following, we prove it for Arikan's $T_2$ binary kernel, and for the $T_3$ kernel given in \eqref{eq-T3Tkernels}. 

\subsection{Examples of kernels polarization} 
In the following, we prove that the kernels $T_2$ and $T_3$ are polarizing according to the multi-kernel definition by showing that the constraints in \eqref{eq-lemma2} are met for these two kernels.

For the kernel $T_2$, consider a channel $\mathcal{W}$ with binary inputs and finite outputs. Let $(U_1,U_2)$ be the Bern($\frac{1}{2}$) auxiliary inputs and $(X_1,X_2)$ be the binary channel inputs where $X_1^2 = U_1^2 \cdot T_2$. 
The independence of $(X_1,X_2)$ implies that $(X_1,Y_1)$ and $(X_2,Y_2)$ are all pairwise independent and $H(X_1| Y_1) = H(X_2|Y_2)= H(X|Y) \triangleq H_0 \leq 1$. 

 The inequality in \eqref{eq-lemma2} implies the pair of constraints: 
\begin{IEEEeqnarray*}{rCl}
| I(U_1; Y_1^2) - I(X; Y)  |   \geq  I(X;Y)^{\alpha} (1 - I(X;Y))^{\beta} && \\
| I(U_2; Y_1^2| U_1) - I(X; Y)  |   \geq I(X;Y)^{\alpha} (1 - I(X;Y))^{\beta},
\end{IEEEeqnarray*}   
but since $I(U_2; Y_1^2| U_1) = 2 I(X; Y) - I(U_1; Y_1^2)$, then it suffices to prove the first inequality, which amounts to
\begin{IEEEeqnarray}{rCl}
 H(X_1 \oplus X_2| Y_1^2) -  H_0 \geq H_0^{\beta}.(1 -  H_0)^{\alpha} \label{eq-target_T2} . 
\end{IEEEeqnarray}
We prove in Appendix \ref{app_ProofT2}, that this holds with the choice $\alpha = 1$ and $\beta =2$.
 
 As for $T_3$, under similar assumptions on $(U_1,U_2,U_3)$ and the independence of $(X_1,Y_1)$, $(X_2,Y_2)$ and $(X_3,Y_3)$, and defining $H(X_1|Y_1) = H(X_2|Y_2) = H(X_3|Y_3) \triangleq H_0$, the condition in \eqref{eq-lemma2} amounts to proving that: 
 \begin{IEEEeqnarray*}{rCl}
| I(U_1; Y_1^3) - I(X; Y)  |   \geq  I(X;Y)^{\alpha} (1 - I(X;Y))^{\beta} && \\
| I(U_2; Y_1^3| U_1) - I(X; Y)  |   \geq I(X;Y)^{\alpha} (1 - I(X;Y))^{\beta} &&\\ 
| I(U_3; Y_1^3| U_1^2) - I(X; Y)  |   \geq I(X;Y)^{\alpha} (1 - I(X;Y))^{\beta},
\end{IEEEeqnarray*}   
Again, it suffices to prove only the two first inequalities, which amounts to proving that
\begin{IEEEeqnarray}{rCl}
 | H(X_1\oplus X_2 \oplus X_3| Y_1^3) - H_0 | & \geq &  H_0 ^{\beta} (1 - H_0 )^{\alpha}   \label{eq-target_T3_1} \\
\hspace*{-1cm} | H(X_2 \oplus X_3| Y_1^3 , X_1 \oplus X_2 \oplus X_3) - H_0 | & \geq & H_0^{\beta} (1 - H_0)^{\alpha}.   \label{eq-target_T3_2} 
\end{IEEEeqnarray} 
In Appendix \ref{app_ProofT3}, we show that this holds with $\alpha = \beta = 2$.

\section{Rate of convergence of multi-kernel polar codes}
The rate of convergence of the sequence $\{I_m\}_m$, which is related to the error exponent of the generator matrix $G_{N} = T_{l_1} \otimes \dots \otimes T_{l_m}$, is the asymptotic convergence rate of the probability of error. 
In this section, we show how to derive the convergence rate of a multi-kernel polar code based on the rate of convergence of each of the constituent kernels $T_{l_1}, \dots , T_{l_m}$.

\subsection{Definitions}
Let us first extend the definition of the Bhattacharyya parameter, a key measure in the rate of convergence of polar codes, to the case of multi-kernel polar codes. 
We recall that the Bhattacharyya parameter associated with a binary input channel $W$ is \vspace{-1mm}
\begin{equation}
 Z(W) \triangleq \sum_{y \in\mathcal{Y}} \sqrt{W(y|1) W(y|0)}. \vspace{-1mm}
\end{equation}

Accordingly to the notation $W^{(N)}_i = W_{b_1, \dots, b_m}$, where $(b_1, \dots, b_m)$ are the mixed radix decomposition of $i$ in the basis $l_1, \dots, l_m$, we define a random Bhattacharyya parameter associated to the random realization of $(B_1,\dots, B_m)$ as follows: 

\begin{definition}[Bhattacharyya parameter]\label{Def-Bhattacharyya}
The random Bhattacharyya parameter associated to the random realization of $(B_1,\dots, B_m)$ is given by
\begin{equation*}
 Z_m \triangleq Z(W_{B_1, \dots, B_m}) = \sum_{z \in\mathcal{Z}} \sqrt{W_{B_1, \dots, B_m}(z|1) W_{B_1, \dots, B_m}(z|0)}. 
\end{equation*}
where $z$ is defined in \eqref{eq-new-channel}.
\end{definition}
The Bhattacharyya parameter is particularly useful since it yields bounds on the block error probability $P_e(N)$ of a length $N$ code, see \cite{exp_urbanke}. 

Next, we define the rate of convergence of a multi-kernel polar code based on the convergence of the sequence $\{Z_m\}_m$. 
\begin{definition} \label{Def-Error-Exp}
A multi-kernel polar code has a rate of convergence $E$ if and only if the following properties hold: 
\begin{enumerate}
\item For all $\gamma \geq E$,  
\begin{equation}
 \lim_{m \to \infty} \mathds{P}( Z_m \geq 2^{-N^\gamma }) = 1;
\end{equation}
\item for all $0 <\gamma \leq E$
\begin{equation}
 \lim_{m \to \infty} \mathds{P}( Z_m \leq 2^{-N^\gamma }) = I(W).
\end{equation}
\end{enumerate}
\end{definition}
The convergence rate relates directly to the error exponent, i.e. the rate of convergence of the block error probability to $0$. As such, if a polar code has an convergence rate of $E$, then the block error probability satisfies
\begin{equation}
 P_e(N) \underset{n\to \infty}{\sim} 2^{-N^E}.
\end{equation}
For polar codes based on Arikan's kernels, i.e. $G_{N}= T_2^{\otimes n}$, the convergence rate was shown in \cite{ArikanRateConvergence} to be equal to $E=0.5$. 
For polar codes formed by larger kernels, authors in \cite{exp_urbanke} derive the rate of convergence through the \emph{partial distances} of the given kernel $T_l$ of length $l$ yielding, for instance, for kernel $T_3$ proposed in (\ref{eq-T3Tkernels}) a rate of convergence of $E=0.42$. 

\subsection{Calculation of the rate of convergence}
The rate of convergence of a polar code can be derived for Arikan's $T_2$ kernels, and more generally, for arbitrary kernels $T_l$, through the following result.

\begin{lemma}{\cite[Theorem 14]{exp_urbanke} }~\\ 
Let $T_l$ be an arbitrary kernel with size $l \times l$, then the rate of convergence of the associated polar code is given by 
\begin{equation}\label{def-E_j}
 E_l \triangleq \dfrac{1}{l} \sum_{i=1}^l \log_l (D_i)
\end{equation}
where $D_i$ is the partial distance of the $i$-th row of the transformation matrix $T_l = (\mathbf{t}^\dagger_1, \dots,\mathbf{t}^\dagger_i, \dots, \mathbf{t}^\dagger_l )^\dagger$, defined as 
\begin{equation}
 D_i \triangleq dist( \mathbf{t}_i, <\mathbf{t}_{i+1},\dots, \mathbf{t}_{l}>)
\end{equation}
where $<\mathbf{t}_{i+1},\dots, \mathbf{t}_{l}>$ is the linear code spanned by the remaining rows of $T_l$ and $A^\dagger$ the transpose of the matrix $A$.
\end{lemma}

In the following, we show that the rate of convergence of a multi-kernel polar code can be derived on the basis of the rate of convergence of each of the constituent kernels.  
\begin{theorem}\label{theo-error-exponents}
Consider a multi-kernel polar code in which each of the $s$ distinct constituent kernels $T_{l_j}$ has an error exponent $E_{l_j}$ for $j \in [1:s]$ and is used with frequency $p_{ j}$ in the Kronecker composition $G_{N}$ as $N \to \infty$. Then, the rate of convergence of the resulting mixed kernel polar code is given by 
\begin{equation}
 E = \sum_{j=1}^s \frac{p_{ j} \log_2(l_j)}{\sum_{j^\prime } p_{ j^\prime } \log_2(l_{j^\prime})} \cdot E_{l_j}
\end{equation}
\end{theorem}
\begin{proof}
Let $m$ be the index of the current iteration in the Kronecker product, let $i \in [1:N]$ and let $(b_1,\dots, b_m)$ be its corresponding mixed radix decomposition. The proof of this theorem follows from the following inequality, which is a result of (\cite[Lemma 13]{exp_urbanke}), 
\begin{IEEEeqnarray}{rCl}
   Z(W_{b_1,\dots, b_{m-1}})^{D_{b_m} } &\leq& Z(W_{b_1,\dots, b_m}) \nonumber \\
   &\leq& 2^{l_m-b_m} Z(W_{b_1,\dots, b_{m-1}})^{D_{b_m}}
\end{IEEEeqnarray}
where $D_{b_m}$, with $b_m \in [1:l_m]$, is the partial distance of the row $\mathbf{t}_{b_m}$ in the kernel $T_{l_m}$. 
The proof that items 1) and 2) in definition \ref{Def-Error-Exp} holds, follows by analytic calculus and is presented in Appendix \ref{app_ProofTheo2}. 
\end{proof} 
%
As a result, the rate of convergence of a multi-kernel polar code consists in a weighted sum of the error exponents of each of the constituent kernels, where the weights are related to the frequency of occurrence of a kernel in the construction. 
 
\appendices 
\section{Proof of polarization of $T_2$ and $T_3$}
 \subsection{Proof of polarization of $T_2$}\label{app_ProofT2}
 
 To prove \eqref{eq-target_T2}, note that 
\begin{IEEEeqnarray}{rCl}
&& H(X_1 \oplus  X_2 | Y_1^2) \nonumber \\
&=& \sum_{y_1, y_2} P_{Y_1Y_2}(y_1,y_2) H(X_1 \oplus X_2 | y_1, y_2) \\
&\overset{(a)}{=}& \sum_{y_1, y_2} P_{Y_1}(y_1) P_{Y_2}(y_2) H(X_1 \oplus X_2 | y_1, y_2) \\
&\overset{(b)}{=}&  \sum_{y_1, y_2} P_{Y_1}(y_1) P_{Y_2}(y_2) h_2(q_{y_1} \star q_{y_2} ) 
\end{IEEEeqnarray}
where $(a)$ is a consequence from the independence of $Y_1$ and $Y_2$, and $(b)$ follows since, conditioned on $(y_1, y_2)$, $X_1 \oplus X_2$ is a binary variable with probability 
\begin{IEEEeqnarray*}{rCl}
 \mathds{P}( X_1\oplus X_2 = 0 |y_1, y_2) &=& \mathds{P}( X_1 = 0 |y_1, y_2) \star \mathds{P}( X_2 = 0 |y_1, y_2) \\
 &=& \mathds{P}( X_1 = 0 |y_1) \star \mathds{P}( X_2 = 0 | y_2) \\
 &\triangleq& q_{y_1}\star q_{y_2} ,  
\end{IEEEeqnarray*}
where $\star$ denotes the binary convolution operator, and where the last step is a result of the memorylessness of the channel. As defined, $q_{y_1}$ and $q_{y_2}$ satisfy the following equality 
\begin{equation*}
 \sum_{y_1} P_{Y_1}(y_1) h_2(q_{y_1}) = \sum_{y_2} P_{Y_2}(y_2) h_2(q_{y_1}) = H(X_1|Y_1) = H_0 .
\end{equation*}

Then, resolting to Mrs Gerber's Lemma \cite{wyner}, i.e., convexity of $h_2(p \star h_2^{-1}(x))$ in $x$, we can write 
\begin{IEEEeqnarray}{rCl}
 H(X_1 \oplus  X_2 | Y_1^2) &\geq& \sum_{y_1} P_{Y_1}(y_1) h_2 \left( q_{y_1} \star h_2^{-1}(H_0) \right) \\
  &\geq&  h_2 \left( h_2^{-1}(H_0) \star h_2^{-1}(H_0) \right) \label{eq-alpha_alpha}
\end{IEEEeqnarray} 

Finally, it can be proved that, $\forall  a  \in [0:1/2] $, 
\begin{equation}
 h_2(a \star a) - h_2(a) \geq h_2^2(a) \cdot (1- h_2(a)) \geq 0   \label{eq-lowerbound_h2_a_star_a}.
\end{equation}
which, when replacing $a = h_2^{-1}(H_0)$, yields the desired result. 

  \subsection{Proof of polarization of $T_3$}\label{app_ProofT3}
  To prove \eqref{eq-target_T3_1}, we first note that similar to \eqref{eq-alpha_alpha}, 
  \begin{IEEEeqnarray}{rCl}
&& H(X_1 \oplus X_2 \oplus X_3| Y_1^3) \nonumber \\
&\geq&  h_2 \left( h_2^{-1}(H_0) \star h_2^{-1}(H_0) \star h_2^{-1}(H_0)\right) \label{eq-alpha_alpha_alpha}\\
 &\overset{(a)}{\geq}& H_0 + h_2 \left( h_2^{-1}(H_0) \star h_2^{-1}(H_0) \right)^2 \cdot(1 - H_0) \\
 &\overset{(b)}{\geq}& H_0 + H_0^2 \cdot (1 - H_0) \geq  H_0 + H_0^2 \cdot(1 - H_0)^2
\end{IEEEeqnarray}
where, $(a)$ and $(b)$ are consequences of \eqref{eq-lowerbound_h2_a_star_a}.

 Next, to prove \eqref{eq-target_T3_2}, we write the following: 
 \begin{IEEEeqnarray*}{rCl}
 && H( X_2 \oplus X_3| Y_1^3, X_1\oplus X_2 \oplus X_3) \nonumber \\
 &=&  H(X_2\oplus X_3| Y_1^3) + H(X_2| Y_1^3) - H(X_1 \oplus X_2 \oplus X_3| Y_1^3) \\
 &=&  H(X_2| Y_2) +  H(X_2 \oplus X_3| Y_1^3) - H(X_1 \oplus X_2 \oplus X_3| Y_1^3) \\
 &=& H_0 +  H(X_2 \oplus X_3| Y_2^3) - H(X_1\oplus X_2 \oplus X_3| Y_1^3) .   \label{eq_intermediate_1} 
 \end{IEEEeqnarray*} 
Next, we note that we can write the following upper bound on $ H(X_2+X_3| Y_2^3)$. 
\begin{IEEEeqnarray}{rCl}
&& H(X_2 \oplus X_3 | Y_2^3) \nonumber \\
&{=}&  \sum_{y_2, y_3} P_{Y_2}(y_2) P_{Y_3}(y_2) h_2(q_{y_2} \star q_{y_3} ) \\
&\overset{(a)} {\leq}& H(X_2|Y_2 ) + (1-H(X_3 |Y_3)) \cdot H(X_2|Y_2)  \\
&=& H_0 + H_0\cdot(1- H_0)
\end{IEEEeqnarray}
where $(a)$ is due to a consequence of Mrs Gerber's lemma \cite{wyner},  
\begin{equation}
 h_2(a\star b) \leq h_2(a) + h_2(b) - h_2(a)h_2(b) . 
\end{equation}
Thus, we can finally write that
 \begin{IEEEeqnarray*}{rCl}
&&  H( X_2 \oplus X_3| Y_1^3, X_1 \oplus X_2 \oplus X_3) - H_0 \nonumber \\
&=& H(X_2 \oplus X_3| Y_2^3) - H(X_1 \oplus X_2 \oplus X_3| Y_1^3) \\
&\overset{(a)} {\leq}& - H_0^2\cdot( 1 -  H(X_2 \oplus X_3| Y_2^3) )  \\ 
&\leq& -  H_0^2\cdot ( 1 -  H_0 - H_0\cdot(1- H_0) ) \\ 
&=&  -  H_0^2 \cdot( 1 -  H_0 )^2   \leq 0 
\end{IEEEeqnarray*} 
where $(a)$ follows similarly from \eqref{eq-alpha_alpha} and \eqref{eq-lowerbound_h2_a_star_a} by leaving $X_2+X_3$ grouped. 
 
 \section{Error exponents: Proof of Theorem \ref{theo-error-exponents}} \label{app_ProofTheo2}
In order to prove that the error exponent of a mutli-kernel polar code is given by 
 \begin{equation}
 E = \sum_{j=1}^s \frac{p_{l_j} \log_2(l_j)}{\sum_{j^\prime } p_{l_{j^\prime}} \log_2(l_{j^\prime})} \cdot E_{l_j}
\end{equation}
we would need to prove that $E$ satisfies both conditions in definition \ref{Def-Error-Exp}, namely 
\begin{itemize}
\item Condition 1):  For all $\gamma \geq E$,  
\begin{equation}
 \lim_{m \to \infty} \mathds{P}( Z_m \geq 2^{-N^\gamma }) = 1;
\end{equation}
\item Condition 2) for all $0 <\gamma \leq E$
\begin{equation}
 \lim_{m \to \infty} \mathds{P}( Z_m \leq 2^{-N^\gamma }) = I(W).
\end{equation}
\end{itemize}

Before proving that both these conditions hold, we list some preliminaries which are essential to the proof. 

\subsection{Preliminaries}
Let $m$ be the current iteration in the multi-kernel construction given by $G_N = \bigotimes_{k=1}^m T_{l_k}$ where the kernels $ T_{l_k}$ can take values in a set of distinct kernels $\{T_1, \dots, T_s\}$ with a probability $p_j$ for $j \in[1:s]$. 
Let $b_1,\dots, b_m $ the realization of the random variables $B_1, \dots, B_m$ up to this iteration. We will use the shorthand notation $Z_m$ to denote the realization of the random Bhattacharryaa parameter $Z(W_{b_1,\dots, b_m})$ defined in Definition \ref{Def-Bhattacharyya}. Let $D_{b_1}, \dots,D_{b_m} $ denote the corresponding partial distances and $N = \prod_{k=1}^m l_j$ be the blocklength at iteration $m$.

\begin{definition}
We define the occurrences set of a kernel $T_{l_j}$ in $G_N$ as 
\begin{IEEEeqnarray}{rCl}
\mathcal{N}_j^m &\triangleq& \{ k \in [1:m], T_k = T_{l_j} \}
\end{IEEEeqnarray} 
and the occurrences set of an index $i \in [1:l_j]$ in the kernel $T_{l_j}$ as
\begin{IEEEeqnarray}{rCl}
\mathcal{N}_{i,j}^m &\triangleq& \{ k \in \mathcal{N}_j^m , b_k = b_i \}, 
\end{IEEEeqnarray} 
with cardinalities
\begin{IEEEeqnarray}{rCl}
n_j^m &\triangleq& ||\mathcal{N}_j^m|| \\
n_{i,j}^m &\triangleq& ||\mathcal{N}_{i,j}^m|| . 
\end{IEEEeqnarray} 
\end{definition}

By the law of large numbers, and since the variables $(B_1, \dots, B_m)$ are each uniformly distributed, the cardinalities of these sets verify
\begin{IEEEeqnarray}{rCl}
&& \lim_{m \to \infty}\dfrac{  n_j^m }{m}  =  p_j \ , \ \lim_{m \to \infty}\dfrac{ n_{i,j}^m}{n_j^m }  =  \dfrac {1}{l_j} \\
  && \qquad \qquad \Rightarrow \lim_{m \to \infty}\dfrac{ n_{i,j}^m}{m}  =  \dfrac { p_j }{l_j}  . 
\end{IEEEeqnarray}
 As such, 
 \begin{equation}\label{Inequality_limit}
 \forall  \epsilon >0,  \exists M \geq 0 , \ \text{s.t.} \  \left| \dfrac{ n_{i,j}^m}{m} - \dfrac { p_j }{l_j}   \right| \leq \epsilon  \ \text{ for all} \  m \geq M . 
\end{equation}

In the following, we list some key properties verified by the sequence of Bhattacharyya parameters $(Z_m)_{m\geq0}$.
\begin{lemma}\label{Lemma-Inequality-Z}
There exists $K\geq 1$, such that, for all $m\geq 1$, 
 \begin{IEEEeqnarray}{rCl}\label{Inequality-Z}
   Z_{m-1}^{D_{b_m} } &\leq& Z_m \leq  K Z_{m-1}^{D_{b_m}} . 
\end{IEEEeqnarray}
\end{lemma}
 \begin{proof}
This inequality follows from simple analytic manipulations of the result of (\cite[Lemma 13]{exp_urbanke}), 
\begin{IEEEeqnarray}{rCl}
  Z_{m-1}^{D_{b_m} }  \leq  Z_m  &\leq& 2^{(l_m-b_m)} Z_{m-1}^{D_{b_m}} \\
   &\leq&  2^{l^\star} Z_{m-1}^{D_{b_m}}  , 
\end{IEEEeqnarray}
where 
\begin{equation}
 l^\star \triangleq \max_{j\in [1:s]} l_j . 
\end{equation}
The results follows then by defining $K \triangleq 2^{l^\star}$.
\end{proof}

\begin{lemma}\label{Lemma-partial-distances}
 For any polarizing kernel $T_{l_j}$, the partial distances $D_i$ for $i \in [1:l_j]$, satisfy the following
\begin{enumerate}
\item All partial distances are greater than $1$ and bounded up by $l_j$
$\forall i \in[1:l_j], \quad  1 \leq D_i \leq l_j$ ; 
\item there exists at least a partial distance $D_{i_j}$, of a row $i_j$ of the kernel's matrix $T_{l_j}$, such that $D_{i_j} \geq 2$ .
\end{enumerate}
\begin{proof}
The first claim follows from the invertibility of a polarizing kernel matrix and the definition of partial distances.
Concerning the second claim, if all partial distances $D_i$ are equal to $1$ for a Kernel $T_l$, \eqref{Inequality-Z} becomes 
\begin{equation}
 Z_{m-1} \leq Z_m  , \text{ for all } m \geq 1 . 
\end{equation}
Thus, $(Z_m)_m$ is a non-decreasing positive sequence, which can converge to $0$ only if it's constant and equal to $0$, contradicting the polarization property. 
\end{proof} 
 \end{lemma}
 The following Lemma will be instrumental in the second part of the proof. 
\begin{lemma}\label{Lemma-Inequality-m1-m2}
For all $m_1$ and $m_2$ such that $m_2 \geq m_1$, we have that 
\begin{equation}\label{Inequality-m1-m2}
   Z_{m_2} \leq \left(K^{(m_2-m_1)} Z_{m_1} \right)^{\displaystyle\prod_{k=m_1+1}^{m_2} D_{b_k}} .
\end{equation}
\end{lemma}
\begin{proof}
Let $m_2 \geq m_1 \geq 0$ be two integers. We have from the right  hand side of inequality \eqref{Inequality-Z} that 
  \begin{IEEEeqnarray}{rCl}
  Z_{m_2} &\leq& K Z_{m_2 -1}^{D_{b_{m_2}}} \\
  & \leq& K  K^{D_{{b_{m_2 }}}} Z_{m_2 -2}^{D_{b_{m_2}} D_{b_{m_2-1}}} \\
  &\leq& K^{  \left[1 + \displaystyle\sum_{m= m_1+2}^{m_2} \prod_{k=m}^{m_2} D_{b_k} \right]} \times  Z_{m_1}^{\displaystyle\prod_{k=m_1+1}^{m_2} D_{b_k}} \\
  &\overset{(a)}{\leq}& K^{  \left[\displaystyle\sum_{m=m_1+1}^{m_2} \prod_{k=m}^{m_2} D_{b_k}\right]} \times Z_{m_1}^{\displaystyle\prod_{k=m_1+1}^{m_2} D_{b_i}} \\
  &\overset{(b)}{ \leq}& K^{ (m_2 -m_1) \displaystyle\prod_{k=m_1+1}^{m_2} D_{b_k}} \times Z_{m_1}^{\displaystyle\prod_{k=m_1+1}^{m_2} D_{b_k}} \\
  &=& \left(K^{(m_2 -m_1)} Z_{m_1} \right)^{\displaystyle\prod_{k=m_1+1}^{m_2} D_{b_k}} , 
\end{IEEEeqnarray}
where $(a)$ and $(b)$ result from both that $K\geq 1$, and from claim 2) Lemma \ref{Lemma-partial-distances}.
\end{proof}
 
\subsection{Proof of condition 1) }\label{proof-part-1}
To show condition 1), we rely on the left hand side of \eqref{Inequality-Z}. 
Given an $ \epsilon \geq 0$ with $m \geq M$ satisfying \eqref{Inequality_limit}, we have that 
\begin{IEEEeqnarray}{rCl}
  Z_m  &\geq  Z_{m-1}^{D_{b_m} } & \geq Z_0^{\displaystyle\prod_{k=1}^m D_{b_m} } . 
\end{IEEEeqnarray}
The exponent of $Z_0$ can be upper bounded as 
 \begin{IEEEeqnarray}{rCl}
\prod_{k=1}^m D_{b_m} &=& \prod_{j=1}^s \prod_{k \in \mathcal{N}_j^m} D_{b_k} \\
&=& \prod_{j=1}^s \prod_{i = 1}^{l_j}  D_{b_i}^{n_{i,j}^m}  \\
&=& \prod_{j=1}^s l_j^{ \left[ \displaystyle\sum_{i = 1}^{l_j} n_{i,j}^m \log_{l_j} \left( D_{b_i} \right) \right] }   \\
&=&\prod_{j=1}^s l_j^{  m \left[   \displaystyle\sum_{i = 1}^{l_j} \dfrac{ n_{i,j}^m}{m} \log_{l_j} \left( D_{b_i }\right) \right] }  \\  
&\overset{(a)}{\leq}&\prod_{j=1}^s l_j^{  m ( p_j + \epsilon^\prime)  \left[  \dfrac{1}{l_j} \displaystyle\sum_{i = 1}^{l_j}  \log_{l_j} \left( D_{b_i }\right) \right] }   \\
&\leq& \prod_{j=1}^s l_j^{  m ( p_j + \epsilon^\prime)  E_{l_j} } \\ 
&\overset{(b)}{\leq}& N^{E_\epsilon } 
\end{IEEEeqnarray}
where $(a)$ stems from \eqref{Inequality_limit}, and $E_{l_j}$ is as defined in \eqref{def-E_j}, and 
\begin{equation}
 \epsilon^\prime \triangleq \epsilon \cdot \displaystyle \max_{j\in[1:s]}l_j  = \epsilon \cdot l^\star ,
\end{equation}
while $(b)$ stems from
\begin{IEEEeqnarray}{rCl}
 N =  \prod_{k=1}^m l_k  =  \prod_{j=1}^s l_j^{n_j^m} \geq  \prod_{j=1}^s l_j^{m (p_j - \epsilon)}\geq  \prod_{j=1}^s l_j^{m (p_j - \epsilon^{\prime})} , 
\end{IEEEeqnarray}
where 
 \begin{IEEEeqnarray}{rCl}
E_\epsilon &\triangleq&  \dfrac{ \sum_{j=1}^s  m ( p_j + \epsilon^{\prime}  )  \log(l_j)  E_{l_j}}{     \sum_{j^\prime =1}^s  m   (p_{j^\prime} - \epsilon^{\prime} ) \log(l_{j^\prime} ) } \\
&=&  \sum_{j=1}^s \dfrac{ ( p_j +  \epsilon^{\prime}) \log(l_j) }{ \sum_{j^ \prime =1}^s  ( p_{j^\prime} -  \epsilon^{\prime})   \log(l_{j^\prime} ) } E_{l_j} \\
&\leq& E + \epsilon^{\dprime} , 
\end{IEEEeqnarray}
with $\epsilon^{\dprime}$ proportional to $\epsilon$. 
Thus, we have that, for $m \geq M$, 
\begin{IEEEeqnarray}{rCl}\label{Eq-FromZto2}
  Z_m 
  &\geq& Z_0^{ N^{E+\epsilon^{\dprime  }  }} \\
  &=& 2^{[N^{E+\epsilon^{\dprime }} \cdot \log_2(Z_0)]} \\
  &=& 2^{ [ - N^{E+\epsilon^{\dprime}} \cdot  N^{\log_N(- \log_2(Z_0))} ]}\\
  &=& 2^{[ - N^{E+\epsilon^{\dprime} + \log_N(- \log_2(Z_0))} ]}\\
  &=& 2^{[ - N^{E+\epsilon^{\trprime } } ]} . 
\end{IEEEeqnarray}
Finally, for all $\gamma \geq E$, 
\begin{equation}
 \lim_{m \to \infty} \mathds{P}( Z_m \geq 2^{-N^\gamma }) = 1 . 
\end{equation}

\subsection{Proof of condition 2)}\label{proof-part-2}
 To prove condition 2), we will rely on the right hand side of \eqref{Inequality-Z} by showing that the exponential decay of $Z_m$ annihilates the role of the multiplicative constant $K$. 
\begin{lemma}\label{Lemma_bound-1}
For every $\epsilon >0$, there exists $M_0$ such that 
\begin{equation}
 \mathds{P} \left( Z_m \leq K^{-(l^\star +1)} \quad \forall m \geq M_0\right) > I_0 -\epsilon  , 
\end{equation}
where $l^\star \triangleq \displaystyle\max_{j\in[1:s]} l_j$.  
\end{lemma}
\begin{proof}
The proof of this Lemma is a generalization of the proof of \cite[Lemma 5.11]{pol_sas}. 
Defining the event 
\begin{IEEEeqnarray}{rCl}
\Omega &\triangleq& \{ w : \lim_{m \to \infty} Z_m(w) = 0 \} \\
&=& \{ w : \forall k > 1, \exists m_0 \ \text{s.t} \     Z_m(w) \leq \dfrac{1}{k} \ \forall m\geq m_0     \} \\
&=& \bigcap_{k > 1} \bigcup_{m_0\geq 0}  A_{m_0,k} , 
\end{IEEEeqnarray}
where 
\begin{equation}
 A_{m_0,k} \triangleq \{ w : Z_m(w) \leq \frac{1}{k} \quad \forall m \geq m_0 \}.  
\end{equation}
 Since the multi-kernel polar code defined by $G_N$ polarizes, i.e. the sequence $(Z_m)_{m > 0}$ converges to a Bern($1- I_0$) random variable, we have that 
\begin{equation}
 \mathds{P} (\Omega) =  I_0. 
\end{equation}
In the following, the notation $Z_m(w)$ will be shortened to $Z_m$ since the dependence on $w$ is implicit. Let $k \geq 1$ be fixed. The sequence of sets $ \left(  A_{m_0,k} \right)_{m_0\geq 0} $ is increasing in $m_0$, thus we can write 
\begin{equation}
  A_{m_0,k} = \bigcup_{ n\leq m_0} A_{n,k} . 
\end{equation}
If we define
\begin{equation}
  A_{\infty,k} \triangleq \bigcup_{ m_0\geq 0} A_{m_0,k} , 
\end{equation}
we can write that 
\begin{IEEEeqnarray}{rCl}
\lim_{m_0\to \infty } \mathds{P} ( A_{m_0,k}) &=& \lim_{m_0\to \infty } \mathds{P} \left(  \bigcup_{ n\leq m_0} A_{n,k} \right) \\
&=& \mathds{P} \left( \lim_{m_0\to \infty }  \bigcup_{ n\leq m_0} A_{n,k} \right) \\
&=&  \mathds{P} \left(  A_{\infty,k} \right) . 
\end{IEEEeqnarray}
Thus, for all $k \geq 1$, for all $\epsilon >0$ there exists $M_0$ such that for all $m_0 \geq M_0$
 \begin{equation}
  \mathds{P} ( A_{m_0,k}) \geq \mathds{P} \left(  A_{\infty,k} \right) -  \epsilon . 
 \end{equation}
More specifically, for $k = K^{ l^\star + 1 }$, for all $\epsilon >0$ there exists $M_0$ such that 
\begin{IEEEeqnarray}{rCl}
  \mathds{P} ( A_{M_0,K^{ l^\star + 1 }}) &\geq& \mathds{P} \left(  A_{\infty,K^{ l^\star + 1 }} \right) -  \epsilon \\
  &\geq& \mathds{P} \left( \bigcap_{k\geq 1} A_{\infty,k} \right) -  \epsilon \\
  &=& \mathds{P} \left( \Omega \right) -  \epsilon \\
  &=& I_0 - \epsilon , 
\end{IEEEeqnarray}
which concludes the proof. 
\end{proof}

The following shows a key inequality on the exponential decay of the sequence $(Z_m)_{m >0}$. 
\begin{lemma}\label{Lemma_bound-2}
For every $\epsilon >0$, there exists $M_1$ such that  for all $m \geq M_1$
\begin{equation}
 \mathds{P} \left( Z_m \leq K^{-\frac{m}{4l^\star} }  \right) > I_0 -\epsilon . 
\end{equation} 
\end{lemma}
\begin{proof}
Fix an $\epsilon > 0$. Following the result of the previous lemma, let $M_0 > 0 $ such that 
\begin{equation}
   \mathds{P} ( A_{M_0,K^{ l^\star + 1 }}) \geq I_0 - \epsilon . 
\end{equation}
Besides, from the right-hand side of inequality \eqref{Inequality-Z}, we have also that 
\begin{IEEEeqnarray}{rCl}
   Z_{m+1} &\leq& K  Z_m^{D_{b_{m+1}}} \\
   &=& K  Z_m^{D_{b_{m+1}}-1} Z_m \\
   &\overset{(a)}{\leq}& K K^{-(l^\star+1)(D_{b_{m+1}}-1)} Z_m\\
   &=& K^{1-(l^\star+1)(D_{b_{m+1}}-1) } Z_m , 
\end{IEEEeqnarray}
where $(a)$ follows from
\begin{equation}\label{Equ-K-lstar}
 Z_m \leq K^{-(l^\star+1)} , 
\end{equation}
for every $w \in  A_{M_0,K^{ l^\star + 1 }}$, for all $m \geq M_0$.
Now, given the result of Lemma \ref{Lemma-partial-distances}, we can bound $Z_{m+1}$ as
\begin{equation}
\left\{
\begin{array}{rcll}
Z_{m+1} &\leq&  Z_m  K^{-l^\star} & \text{ if }  b_{m+1} = i_{m+1} \\
      Z_{m+1} &\leq& Z_m  K & \text{ if }  b_{m+1} \neq i_{m+1}    
\end{array} \right.
\end{equation}
where $i_{m+1} $ is the index for which the partial distance satisfies $D_{i_{m+1}}\geq 2$.

Given two integers $m_0,m_1$ such that $m_1 \geq m_0 \geq M_0$,
\begin{IEEEeqnarray}{rCl}
   Z_{m_1} &\leq&  Z_{m_0} K^{(m_1-m_0)(1 - \alpha_{m_0}^{m_1}(1+ l^\star))} , 
\end{IEEEeqnarray}
where $\alpha_{m_0}^{m_1}$ is the fraction of occurrences of the indices $i_j$ between iterations $m_0$ and $m_1$, i.e., 
\begin{equation}
 \alpha_{m_0}^{m_1} \triangleq \dfrac{|| k\in [m_0:m_1], b_k = i_k||}{m_1-m_0} . 
\end{equation}

If we define the typical set $T_{\alpha,m_0}^{m_1}$ as the set of events for which each index $b_j$ occurs at least $\alpha p_j$ fraction of the time between $m_0$ and $m_1$, we have that

\begin{equation}\label{Equ-inequality-alpha-m0m1}
 \alpha_{m_0}^{m_1} \geq \sum_{j=1}^s p_j \alpha = \alpha  
\end{equation}
inside $T_{\alpha,m_0}^{m_1}$. Note that the typical sets $T_{\alpha,m_0}^{m_1}$ are non-increasing in $\alpha$ for a fixed $m_0$ and $m_1$, and increasing in $m_1$ for a fixed $\alpha$ and $m_0$.  
Given
 \begin{equation}
  \alpha^\star  = \dfrac{2 l^\star + 1}{2 (l^\star + 1)} \dfrac{1}{l^\star} ,
 \end{equation}
for $\alpha^\star$ it holds that
 \begin{equation}
  \alpha^\star  \leq  \dfrac{1}{l^\star}   \leq \dfrac{1}{l_j}  \text{ for all } j\in [1:s]  . 
 \end{equation} 
Thus, since the sets $T_{\alpha,m_0}^{m_1}$ are non-increasing in $\alpha$, we can write that
\begin{equation} \label{Equ-limit-Typicalset-alpha}
 \lim_{m_1\to\infty} \mathds{P} (T_{\alpha^\star,m_0}^{m_1}) \geq  \lim_{m_1\to\infty} \mathds{P} (T_{\frac{1}{l^\star},m_0}^{m_1}) . 
 \end{equation} 
We now define the typical set $T_{m_0}^{m_1}$ as the set of events for which each index $b_j \in [1:l_j]$ appears with a probability at least equal to $\frac{p_j}{l_j}$. 

Since every index $b_j$ occurs asymptotically with a probability $\frac{p_j}{l_j}$, then
\begin{equation}
 \lim_{m_1\to\infty} \mathds{P} (T_{m_0}^{m_1}) = 1 , 
\end{equation}
hence, and for all $\epsilon > 0 $ there exists an $M_1 \geq M_0$ such that, for all $m_1 \geq M_1$, 
\begin{equation}\label{Equ-limit-Typicalset}
 \mathds{P} (T_{m_0}^{m_1}) \geq 1 - \epsilon . 
\end{equation}

Next, since $T_{m_0}^{m_1} \subset T_{\frac{1}{l^\star},m_0}^{m_1}$, we have that 
\begin{equation}\label{Equ-limit-Typicalset2}
 \mathds{P} (T_{\frac{1}{l^\star},m_0}^{m_1}) \geq \mathds{P} (T_{m_0}^{m_1}).
\end{equation}
Now, combining \eqref{Equ-limit-Typicalset-alpha} and \eqref{Equ-limit-Typicalset2}, for all $\epsilon > 0 $ there exists an $M_1 \geq M_0$ such that, for all $m_1 \geq M_1$, 
\begin{equation}
 \mathds{P}  (T_{\alpha^\star,m_0}^{m_1}) \geq 1 - \epsilon .
\end{equation}
For such $\epsilon$, let us choose $M_1$ such that $M_1 \geq 2 M_0$, so that, for all $m_1 \geq M_1$, inside the set $T_{\alpha^\star,M_0}^{m_1} \cap A_{M_0,K^{ l^\star + 1 }}$, 
\begin{IEEEeqnarray}{rCl}
   Z_{m_1} &\leq&  Z_{M_0} K^{(m_1-M_0)(1 - \alpha_{M_0}^{m_1}(1+ l^\star))} \\
   &\overset{(a)}{\leq} &Z_{M_0} K^{(m_1-M_0)(1 - \alpha^\star (1+ l^\star))} \\
   &\overset{(b)}{\leq}&  Z_{M_0} K^{ -\frac{ m_1}{4 l^\star}} \\
   &\overset{(c)}{\leq}& K^{ -\frac{ m_1}{4 l^\star}} , 
\end{IEEEeqnarray}
where $(a)$ follows from \eqref{Equ-inequality-alpha-m0m1}, $(b)$ follows from the assumption $ m_1 \geq M_1 \geq 2 M_0$ and the definition of $ \alpha^\star$ as
\begin{equation}
 m_1 - M_0 \geq \dfrac{m_1}{2} \text{ and  } 1 - \alpha^\star(1+ l^\star) = - \dfrac{1}{2l^\star}  ,  
\end{equation}
and $(c)$ stems from $Z_{M_0} \leq 1$. 
The lemma is now proved noting that the probability of the set $T_{\alpha^\star,M_0}^{m_1} \cap A_{M_0,K^{ l^\star + 1 }}$ satisfies  
 \begin{IEEEeqnarray}{rCl}
  \mathds{P} \left( T_{\alpha^\star,M_0}^{m_1} \cap A_{M_0,K^{ l^\star + 1 }} \right) &\geq& (I_0 - \epsilon)(1 - \epsilon) \\
  &\geq& I_0 - 2 \epsilon . 
\end{IEEEeqnarray}
\end{proof}

Now we are ready to prove condition 2) of definition \ref{Def-Error-Exp}. Let $\epsilon >0$ and let $\alpha <\frac{1}{l^\star}$, and let $\gamma <1$ such that
\begin{equation}\label{Equ-definition-gamma}
  \forall j\in[1:s]  \ \ \alpha l_j  \gamma > 1 - \epsilon .
\end{equation}
We define the constant $C$, independent of the blocklength, as 
\begin{equation}
  C \triangleq \prod_{j=1}^s l_j^{ p_j l_j E_{l_j}} . 
\end{equation}
Let $m_3$ be an integer large enough so that, by defining
\begin{IEEEeqnarray}{rCl}
  m_1 &\triangleq& \dfrac{\log(2m_3)8 l^\star}{\alpha \log(C)} \label{Equ-def-m1} \\ 
  m_2 &\triangleq& \left( 1 + \dfrac{1}{8 l^\star} \right)  m_1 \label{Equ-def-m2}  
\end{IEEEeqnarray}
it holds that
\begin{IEEEeqnarray}{rCl}
  m_1 &\geq& \max (M_0, 8 l^\star)  \label{constraint-m1}\\ 
  m_3 - m_2 &\geq& \gamma m_3 \label{constraint-m3m2}  , 
\end{IEEEeqnarray}
while the typical sets $T_{\alpha, m_1}^{m_2}$, $T_{\alpha, m_2}^{m_3}$ verify 
 \begin{IEEEeqnarray}{rCl}
 \mathds{P} \left(  T_{\alpha, m_1}^{m_2} \right) > 1 - \epsilon \text{ and }  \mathds{P} \left(  T_{\alpha, m_2}^{m_3} \right) > 1 - \epsilon \ . 
\end{IEEEeqnarray}

We start the proof by bounding $Z_{m_2}$:  
\begin{IEEEeqnarray}{rCl}
  Z_{m_2} &\overset{(a)}{\leq}&  ( K^{m_2 -m_1} Z_{m_1} )^{\displaystyle\prod_{i=m_1+1}^{m_2} D_{b_i}} \\
  &\overset{(b)}{\leq}& ( K^{m_2 -m_1 - \frac{m_1}{4 l^\star}} )^{\displaystyle\prod_{i=m_1+1}^{m_2} D_{b_i}} \\
  &\overset{(c)}{ =}& ( K^{-  \frac{m_1}{8 l^\star}} )^{\displaystyle\prod_{i=m_1+1}^{m_2} D_{b_i}} \\
  &\overset{(d)}{\leq}&  K^{- \displaystyle\prod_{i=m_1+1}^{m_2} D_{b_i}} \\
  &\overset{(e)}{\leq}&  K^{- \displaystyle\prod_{j=1}^{s} \prod_{m=1}^{l_j}  D_{b_m} ^{\alpha p_j(m_2 - m_1)}}  \\
  &=&  K^{- \displaystyle\prod_{j=1}^{s}  l_j^{E_{l_j} l_j \alpha p_j(m_2 - m_1)}}  \\
  &=&  K^{- C^{ \alpha (m_2 - m_1)}} \\
  &\overset{(f)}{=}& K^{- C^{ \alpha \frac{m_1}{8 l^\star}}} \label{Equ-BoundZ2}, 
\end{IEEEeqnarray}
where $(a)$ stems from Lemma \ref{Lemma-Inequality-m1-m2}, $(b)$ is the result of Lemma \ref{Lemma_bound-2}, $(c)$ results from the definition of $m_2$ in \eqref{Equ-def-m2}, $(d)$ is a consequence of the constraint \eqref{constraint-m1}, $(e)$ stems from the definition of $T_{\alpha, m_1}^{m_2}$, and $(f)$ is a result of \eqref{Equ-def-m2}.
Let us now bound the term $Z_{m_3}$ as 
\begin{IEEEeqnarray}{rCl}
  Z_{m_3} &\leq&  ( K^{m_3 -m_2} Z_{m_2} )^{\displaystyle\prod_{i=m_2+1}^{m_3} D_{b_i}} \\
  & \overset{(a)}{\leq}&  ( K^{m_3 - C^{ \alpha  \frac{m_1}{8}} } )^{\displaystyle\prod_{i=m_2+1}^{m_3} D_{b_i}} \\
  &\overset{(b)}{=}& ( K^{-  \frac{m_3}{2}} )^{\displaystyle\prod_{i=m_2+1}^{m_3} D_{b_i}} \\
  &\overset{(c)}{\leq}&  K^{- \displaystyle\prod_{i= m_2+1}^{m_3} D_{b_i}} \\
  &\overset{(d)}{\leq} &  K^{- \displaystyle\prod_{j=1}^{s} \prod_{m=1}^{l_j}  D_{b_m} ^{\alpha p_j(m_3 - m_2)}}  \\
  &=&  K^{- \displaystyle\prod_{j=1}^{s}  l_j^{E_{l_j} l_j \alpha p_j(m_3 - m_2)}}  \\
   &\overset{(e)}{\leq}&  K^{- \displaystyle\prod_{j=1}^{s}  l_j^{E_{l_j} l_j \alpha p_j \gamma m_3 }}  \\
   &\overset{(f)}{\leq}& K^{- \displaystyle\prod_{j=1}^{s}  l_j^{E_{l_j}p_j (1 - \epsilon) m_3 }}  \\
   &\overset{(g)}{=}& K^{-N_3^{E(1-\epsilon)}} , 
\end{IEEEeqnarray}
where $(a)$ is a consequence of \eqref{Equ-BoundZ2}, $(b)$ stems from the definition of $m_1$ in \eqref{Equ-def-m1}, $(c)$ results from the not-limiting assumption $m_3 > 2$, $(d)$ follows from the definition of $T_{\alpha, m_2}^{m_3}$, $(e)$ follows from the assumption in \eqref{constraint-m3m2}, $(f)$ is a consequence of \eqref{Equ-definition-gamma}, while $(g)$ follows by noting that the block length at iteration $m_3$ is $N_3 = \prod_{j=1}^{m_3} l_j$. 
The proof is concluded using similar calculations as in \eqref{Eq-FromZto2}, and noticing that
 \begin{equation}
   \mathds{P} \left( \{w: Z_m \leq K^{-\frac{m}{4l^\star}}\}\cap T_{\alpha, m_1}^{m_2}\cap T_{\alpha, m_2}^{m_3}\right) \geq I_0 - \epsilon. 
  \end{equation} 
 
\bibliographystyle{IEEEtran}
\bibliography{polar_codes_bib}
 
\end{document}